\documentclass[showpacs,pra,superscriptaddress,10pt,twocolumn]{revtex4}%
\usepackage{amsfonts}
\usepackage{amsmath}
\usepackage{amssymb}
\usepackage{graphicx}%
\newtheorem{theorem}{Theorem}

\newtheorem{corollary}[theorem]{Corollary}

\newenvironment{proof}[1][Proof]{\noindent\textbf{#1.} }{\ \rule{0.5em}{0.5em}}
\begin{document}
\title{Exploration of nonlocalities in ensembles consisting of bipartite quantum states}
\author{Ming-Yong Ye}
\affiliation{Department of Physics and Center of Theoretical and Computational Physics,
University of Hong Kong, Pokfulam Road, Hong Kong, China}
\affiliation{School of Physics and Optoelectronics Technology, Fujian Normal University,
Fuzhou 350007, China}
\author{Yan-Kui Bai}
\affiliation{College of Physical Science and Information Engineering and Hebei Advance Thin
Films Laboratory, Hebei Normal University, Shijiazhuang, Hebei 050016, China}
\author{Xiu-Min Lin}
\affiliation{School of Physics and Optoelectronics Technology, Fujian Normal University,
Fuzhou 350007, China}
\author{Z. D. Wang}
\affiliation{Department of Physics and Center of Theoretical and Computational Physics,
University of Hong Kong, Pokfulam Road, Hong Kong, China}

\begin{abstract}
It is revealed that ensembles consisting of multipartite quantum states can
exhibit different kinds of nonlocalities. An operational measure is introduced
to quantify nonlocalities in ensembles consisting of bipartite quantum states.
Various upper and lower bounds for the measure are estimated and the exact
values for ensembles consisting of mutually orthogonal maximally entangled
bipartite states are evaluated.

\end{abstract}

\maketitle

\textit{Introduction.}--Although mutually orthogonal multipartite quantum
states are always distinguishable through joint measurements, it was found by
Bennett \textit{et al. }in 1999 that there are ensembles of mutually
orthogonal bipartite product states that cannot be distinguished by means of
local operations and classical communication (LOCC); this phenomenon was
referred to as the nonlocality without entanglement~\cite{bennett}. The
essence of this nonlocality is that the maximal information achievable through
LOCC to distinguish an ensemble of mutually orthogonal multipartite states may
be strictly less than that could be obtained through joint measurements. Since
the discovery of this phenomenon, substantial efforts have been devoted to
search the conditions under which a given ensemble can exhibit such a kind of
nonlocality \cite{ye,cc1,c1,c2,c3,c4,c5,c6,c7,4bell,2state,cohen}. As is
known, any ensemble consisting of the four Bell states can exhibit the
nonlocality \cite{4bell}, while the ensembles consisting of only two
orthogonal states cannot~\cite{2state}. As a latest result, it was indicated
recently that the nonlocality exists in almost all ensembles consisting of
more than $d$ mutually orthogonal $d^{\otimes n}$ states~\cite{cohen}.

So far, most investigations on this intriguing ensemble nonlocality have been
made on qualitative descriptions in analogy to quantum entanglement judgment
(i.e., finding conditions that can be used to check whether a given ensemble
can exhibit the nonlocality), while the quantification of the nonlocality and
its implications are still awaited. To our knowledge, quantification of the
nonlocalities in ensembles has only been addressed in detail in Ref.
\cite{quanti ensem}, though how to quantify quantum entanglement has been
intensively studied.

In this paper we explore nonlocalities in ensembles consisting of bipartite
quantum states and introduce an operational measure for them. The measure is
defined through considering the tensor power\ $\varepsilon^{\otimes n}$,
rather than the ensemble $\varepsilon$ itself, such that some results from
information theory can be used directly. Our discussions focus mainly on
ensembles whose states are mutually orthogonal, motivated by the following two
questions: if an ensemble of mutually orthogonal bipartite states cannot be
distinguished by LOCC, how much entanglement in addition to LOCC is needed
\cite{cohen2}? and if they can be distinguished by LOCC, how much entanglement
can be distilled in the process of locally distinguishing?

\textit{Entanglement charge.--}LOCC distinguishing the states in the ensemble
$\varepsilon=\left\{  p_{X},\rho_{X}^{AB}\right\}  $ can be conceived as a
game. Suppose that there is a classical information source producing symbol
$X$ with probability $p_{X}$. If the source outputs symbol $X$, Alice and Bob
will be given a quantum state $\rho_{X}^{AB}$. They know the ensemble
$\varepsilon$ and their task is to determine the value of $X$ via a
measurement implemented through LOCC. How much information they have gained
about the value of $X$ can be described by the mutual information between $X$
and the measurement result $Y$,
\begin{equation}
I\left(  X;Y\right)  =H\left(  X\right)  +H\left(  Y\right)  -H\left(
XY\right)  ,
\end{equation}
where $H\left(  \cdot\right)  $ is the Shannon entropy of the random
variables. The value of $X$ can be determined through the measurement result
$Y$ if and only if $I\left(  X;Y\right)  =H\left(  X\right)  $ \cite{book}.
The maximal mutual information achievable through LOCC is called locally
accessible information and it will be denoted by $I^{LOCC}\left(
\varepsilon\right)  $. Similarly we can define $I^{Global}\left(
\varepsilon\right)  $ which is the maximal mutual information achievable
through joint measurements and will be equal to $H\left(  X\right)  $ when the
states are mutually orthogonal. Generally there is $I^{LOCC}\left(
\varepsilon\right)  \leq I^{Global}\left(  \varepsilon\right)  \leq H\left(
X\right)  $.

The tensor power of the ensemble $\varepsilon=\left\{  p_{X},\rho_{X}%
^{AB}\right\}  $ is denoted as $\varepsilon^{\otimes n}=\left\{  p_{X^{n}%
},\rho_{X^{n}}^{A^{n}B^{n}}\right\}  $, where $p_{X^{n}}=p_{X_{1}}p_{X_{2}%
}\cdots p_{X_{n}}$, $\rho_{X^{n}}^{A^{n}B^{n}}=\rho_{X_{1}}^{A_{1}B_{1}%
}\otimes\rho_{X_{2}}^{A_{2}B_{2}}\cdots\otimes\rho_{X_{n}}^{A_{n}B_{n}}$ and
$X_{i}$ are independent and identically distributed classical variables as
$X$. Now Alice holds $A^{n}$ and Bob holds $B^{n}$. To obtain the information
about the value of $X^{n}$, they make a measurement that satisfies the
conditions: (1) the mutual information between $X^{n}$ and the measurement
result $Y$ satisfies $I\left(  X^{n};Y\right)  \geq I^{Global}\left(
\varepsilon^{\otimes n}\right)  -\delta_{n}$ with $\lim_{n\rightarrow\infty
}\delta_{n}=0$; (2) it is implemented through LOCC plus $n\times\alpha_{n}$
ebits of entanglement; (3) when the measurement result $Y$ with the
probability $p_{Y}$ is obtained, $n\times\beta_{nY}$ ebits of entanglement is
distilled. We now introduce a new quantity--entanglement charge, which is
defined as%
\begin{equation}
N\left(  \varepsilon\right)  =\inf\lim_{n\rightarrow\infty}\left(  \alpha
_{n}-\sum_{Y}p_{Y}\times\beta_{nY}\right)  ,
\end{equation}
where the infimum operation is taken over all measurements satisfying the
above conditions.
$N\left(  \varepsilon\right)  $ is applicable to ensembles consisting of
general bipartite states and its value can be positive, negative, and zero.

\textit{Two kinds of ensemble nonlocalities and their quantification}-The
ensembles with positive $N\left(  \varepsilon\right)  $ are different from
those with negative $N\left(  \varepsilon\right)  $ in the sense that they can
exhibit different kinds of nonlocalities. We focus on ensembles consisting of
mutually orthogonal bipartite states hereafter. Now there is $I^{Global}%
\left(  \varepsilon^{\otimes n}\right)  =H\left(  X^{n}\right)  $ and the
first condition required for the measurement can be interpreted that the
states in the ensemble $\varepsilon$ will be distinguished with vanishing error.

In the case $N\left(  \varepsilon\right)  >0$, quantum entanglement is needed
in addition to LOCC to distinguish the states in the ensemble $\varepsilon$
with vanishing error. The meaning of the positive entanglement charge
$N\left(  \varepsilon\right)  $ can be manifested through the symbolic
expression%
\begin{equation}
N\left(  \varepsilon\right)  \left[  qq\right]  +LOCC\left\vert _{\varepsilon
}\right.  =>I^{Global}\left(  \varepsilon\right)  , \label{e3}%
\end{equation}
where $\left[  qq\right]  $ means an ebit of quantum entanglement and
$I^{Global}\left(  \varepsilon\right)  =H\left(  X\right)  $. Noting that
$H\left(  X\right)  $ is the information needed to distinguish the states in
$\varepsilon$, the positive $N\left(  \varepsilon\right)  $ quantifies the
minimal nonlocal resources (entanglement) that is needed asymptotically in
addition to LOCC to get the full information $H\left(  X\right)  $ to
distinguish the states in $\varepsilon$. In this case we refer to the
corresponding nonlocality exhibited by ensembles as information nonlocality
and employ $N\left(  \varepsilon\right)  $ as a measure to quantify it.

In the case $N\left(  \varepsilon\right)  <0$, quantum entanglement is not
needed in addition to LOCC to distinguish the states in the ensemble
$\varepsilon$ with vanishing error (if the entanglement is still needed to
assist the process, it could be viewed as a kind of catalyst), and
additionally some entanglement can be distilled. Similarly the meaning of the
negative entanglement charge $N\left(  \varepsilon\right)  $ can be manifested
through the expression
\begin{equation}
LOCC\left\vert _{\varepsilon}\right.  =>I^{Global}\left(  \varepsilon\right)
+\left\vert N\left(  \varepsilon\right)  \right\vert \left[  qq\right]  .
\end{equation}
In this case, the ensemble $\varepsilon$ has no information nonlocality,
however it still exhibits some kind of nonlocality since certain entanglement
can be distilled. Hereafter we may refer to such kind of ensemble nonlocality
as entanglement nonlocality, and employ $\left\vert N\left(  \varepsilon
\right)  \right\vert $ as its measure since it quantifies the maximal quantum
entanglement that can be distilled in the asymptotic limit.

Interestingly, in the case $N\left(  \varepsilon\right)  =0$, the ensemble
$\epsilon$ has neither information nonlocality nor entanglement nonlocality
mentioned above. As a typical example, $N\left(  \varepsilon\right)  =0$ for
ensembles that consist of LOCC distinguishable product states \cite{add}.

\textit{Bounds for entanglement charge.}--Although the entanglement charge
$N\left(  \varepsilon\right)  $ is usually hard to compute, some useful bounds
of it can be obtained. Bennett \textit{et al.} considered how much quantum
transmission is needed to complete a special distinguishing measurement in
\cite{bennett}, where an upper bound for our defined $N\left(  \varepsilon
\right)  $ is implied, i.e., $N\left(  \varepsilon\right)  \leq S\left(
\rho^{A}\right)  $ where $\rho^{A}=Tr_{B}\left(  \sum_{X}p_{X}\rho_{X}%
^{AB}\right)  $ and $S\left(  \cdot\right)  $ is the quantum entropy. This can
be obtained through the protocol that Alice first compresses her state
\cite{compress} and teleports it to Bob \cite{teleport} and Bob then
distinguishes the states locally. The following are more tight upper bounds.

\begin{theorem}
Suppose $\varepsilon=\left\{  p_{X},\rho_{X}^{AB}\right\}  $ is an ensemble
consisting of mutually orthogonal bipartite states. The entanglement charge
will satisfy
\begin{align}
N\left(  \varepsilon\right)   &  \leq S\left(  A\left\vert B\right.  \right)
=S\left(  \rho^{AB}\right)  -S\left(  \rho^{B}\right)  ,\label{e4}\\
N\left(  \varepsilon\right)   &  \leq S\left(  B\left\vert A\right.  \right)
=S\left(  \rho^{AB}\right)  -S\left(  \rho^{A}\right)  , \label{e5}%
\end{align}
where $\rho^{AB}=\sum_{X}p_{X}\rho_{X}^{AB}$, $\rho^{B}=Tr_{A}\rho^{AB}$,
$\rho^{A}=Tr_{B}\rho^{AB}$ and $S\left(  \cdot\right)  $ is the quantum entropy.
\end{theorem}

\begin{proof}
The theorem can be derived from the quantum state merging \cite{merge1,merge2}%
. To distinguish the states in the ensemble $\varepsilon^{\otimes n}$, the
part $A^{n}$ on Alice's side can be merged to Bob and then he distinguishes
the states locally. In the process of merging \cite{merge1,merge2}, the net
consumed entanglement can be $S\left(  A\left\vert B\right.  \right)  $ ebits,
so Eq. (\ref{e4}) is obtained. If the part $B^{n}$ on Bob's side is first
merged to Alice and then she distinguishes the states locally, similarly the
net consumed entanglement can be $S\left(  B\left\vert A\right.  \right)  $
ebits, so Eq. (\ref{e5}) is obtained.
\end{proof}

The above upper bounds depend only on the state $\rho^{AB}$, so different
ensembles may have the same upper bound. It is possible to get a smaller bound
if we examine the states in the ensemble carefully since it is possible that
only a part of $A$ needs to be merged to Bob and then the states become LOCC distinguishable.

\begin{theorem}
Suppose $\varepsilon=\left\{  p_{X},\rho_{X}^{AB}\right\}  $ is an ensemble
consisting of mutually orthogonal bipartite pure states. The entanglement
charge will satisfy%
\begin{equation}
N\left(  \varepsilon\right)  \geq\sum p_{X}S\left(  \rho_{X}^{A}\right)
-I_{\rho^{AB}}\left(  A;B\right)  , \label{tt}%
\end{equation}
where $S\left(  \cdot\right)  $ is the quantum entropy and $I_{\rho^{AB}%
}\left(  A;B\right)  =S\left(  \rho^{A}\right)  +S\left(  \rho^{B}\right)
-S\left(  \rho^{AB}\right)  $ is the quantum mutual information with
$\rho^{AB}=\sum_{X}p_{X}\rho_{X}^{AB}$, $\rho^{A}=Tr_{B}\rho^{AB}$, $\rho
^{B}=Tr_{A}\rho^{AB}$ and $\rho_{X}^{A}=Tr_{B}\rho_{X}^{AB}$.
\end{theorem}

The quantum mutual information $I_{\rho^{AB}}\left(  A;B\right)  $ is always
nonnegative and it can be regarded as a quantification of the total
correlation between $A$ and $B$, so Eq. (\ref{tt}) means that the entanglement
charge is not smaller than the average entanglement of the states in the
ensemble minus the total correlation between $A$ and $B$. When $\rho^{AB}%
=\rho^{A}\otimes\rho^{B}$ there is $I_{\rho^{AB}}\left(  A;B\right)  =0$ and
$N\left(  \varepsilon\right)  \geq\sum p_{X}S\left(  \rho_{X}^{A}\right)  $;
the average entanglement of the states in the ensemble is a lower bound of the
entanglement charge. This case will happen when we consider an ensemble
consisting of a full basis states with equal probability.

\begin{proof}
According to the definition of $N\left(  \varepsilon\right)  $ we should
consider distinguishing the states of the ensemble $\varepsilon^{\otimes
n}=\left\{  p_{X^{n}},\rho_{X^{n}}^{A^{n}B^{n}}\right\}  $ using LOCC plus
$n\times\alpha_{n}$ ebits of entanglement. It is equivalent to distinguish the
states of the ensemble $\left\{  p_{X^{n}},\rho_{X^{n}}^{A^{n}B^{n}}%
\otimes\Phi_{n}^{A_{0}B_{0}}\right\}  $ using LOCC only, where $\Phi
_{n}^{A_{0}B_{0}}$ is a bipartite pure state with $S\left(  \Phi_{n}^{A_{0}%
}\right)  =S\left(  \Phi_{n}^{B_{0}}\right)  =n\times\alpha_{n}$. The mutual
information between $X^{n}$ and the measurement result $Y$ will satisfy
\cite{badziag,upper}
\begin{align}
I\left(  X^{n};Y\right)   &  \leq n\left(  S\left(  \rho^{B}\right)  +S\left(
\rho^{A}\right)  -\sum p_{X}S\left(  \rho_{X}^{A}\right)  \right) \nonumber\\
&  +n\left(  \alpha_{n}-\sum p_{Y}\beta_{nY}\right)  .
\end{align}
Noting that $I\left(  X^{n};Y\right)  \geq I^{Global}\left(  \varepsilon
^{\otimes n}\right)  -\delta_{n}$ is required and there is $I^{Global}\left(
\varepsilon^{\otimes n}\right)  =nS\left(  \rho^{AB}\right)  $, we can get%
\begin{equation}
\alpha_{n}-\sum p_{Y}\beta_{nY}\geq\sum p_{X}S\left(  \rho_{X}^{A}\right)
-I_{\rho^{AB}}\left(  A;B\right)  -\delta_{n}/n. \label{er}%
\end{equation}
Since $\delta_{n}/n$ will go to zero in the asymptotic limit, we can get Eq.
(\ref{tt}) from Eq. (\ref{er}).
\end{proof}

The above two theorems give upper and lower bounds on $N\left(  \varepsilon
\right)  $. It is valuable to know when the upper and the lower bounds will be
close. We first rewrite the lower bound expression in theorem 2 as%
\begin{equation}
N\left(  \varepsilon\right)  \geq S\left(  A\left\vert B\right.  \right)
-\chi^{A}\left(  \varepsilon\right)  =S\left(  B\left\vert A\right.  \right)
-\chi^{B}\left(  \varepsilon\right)  ,
\end{equation}
where $\chi^{A}\left(  \varepsilon\right)  =S\left(  \rho^{A}\right)  -\sum
p_{X}S\left(  \rho_{X}^{A}\right)  $ and $\chi^{B}\left(  \varepsilon\right)
=S\left(  \rho^{B}\right)  -\sum p_{X}S\left(  \rho_{X}^{B}\right)  $ are the
Holevo information of the ensembles seen by Alice and Bob, respectively. As
for ensembles consisting of orthogonal pure states, it is not hard to find
that the contents of the two theorems can be summarized as
\begin{equation}
S\left(  A\left\vert B\right.  \right)  -\chi^{A}\left(  \varepsilon\right)
\leq N\left(  \varepsilon\right)  \leq S\left(  A\left\vert B\right.  \right)
\label{r1}%
\end{equation}
when $\chi^{A}\left(  \varepsilon\right)  \leq\chi^{B}\left(  \varepsilon
\right)  $ or
\begin{equation}
S\left(  B\left\vert A\right.  \right)  -\chi^{B}\left(  \varepsilon\right)
\leq N\left(  \varepsilon\right)  \leq S\left(  B\left\vert A\right.  \right)
\label{r2}%
\end{equation}
when $\chi^{B}\left(  \varepsilon\right)  \leq\chi^{A}\left(  \varepsilon
\right)  $. From Eqs. (\ref{r1}) and (\ref{r2}), we know that the upper and
the lower bounds will be closer if ever $\chi^{A}\left(  \varepsilon\right)  $
or $\chi^{B}\left(  \varepsilon\right)  $ is smaller. Noting that $\chi
^{A}\left(  \varepsilon\right)  $ is the upper bound of the information about
the value of $X$ that Alice can gain solely \cite{book}, Eq. (\ref{r1}) means
that the difference between the bounds will be small if Alice can gain little
information about the value of $X$ without cooperation with Bob. When
$\chi^{A}\left(  \varepsilon\right)  =0$ or $\chi^{B}\left(  \varepsilon
\right)  =0$, the exact value of $N\left(  \varepsilon\right)  $ can be
obtained. This occurs only when all $\rho_{X}^{A}$ (or $\rho_{X}^{B}$) are the
same. Consequently, we have the following result.

\begin{corollary}
Suppose that $\varepsilon=\left\{  p_{X},\rho_{X}^{AB}\right\}  $ is an
ensemble consisting of mutually orthogonal $d\times d$ maximally entangled
pure states. The entanglement charge will be%
\begin{equation}
N\left(  \varepsilon\right)  =S\left(  \rho^{AB}\right)  -S\left(  \rho
^{B}\right)  =H\left(  X\right)  -\log d, \label{rt}%
\end{equation}
where $H\left(  \cdot\right)  $ is the Shannon entropy, $\rho^{AB}=\sum
_{X}p_{X}\rho_{X}^{AB}$, and $\rho^{B}=Tr_{A}\rho^{AB}$.
\end{corollary}

The corollary is true since all $\rho_{X}^{A}=Tr_{B}\rho_{X}^{AB}$ and
$\rho_{X}^{B}=Tr_{A}\rho_{X}^{AB}$ are the same maximally mixed state, so
there are both $\chi^{A}\left(  \varepsilon\right)  =0$ and $\chi^{B}\left(
\varepsilon\right)  =0$; the upper and the lower bounds of $N\left(
\varepsilon\right)  $ becomes the same value. As is known, the four Bell
states cannot be distinguished by LOCC \cite{4bell}, however the corollary
indicates that the entanglement charge $N\left(  \varepsilon\right)  $ of an
ensemble consisting of these four states can be any value between $-1$ and $1$
dependent on their probabilities, so it may have the information nonlocality,
entanglement nonlocality or neither. The reason is that whether the states in
the ensemble $\varepsilon$ are LOCC indistinguishable depends only on its
states, while whether the ensemble has the information nonlocality or the
entanglement nonlocality depends not only on its states but also on the
probabilities of the states.

There are other ways to obtain upper bounds for $N\left(  \varepsilon\right)
$. For any ensemble $\varepsilon=\left\{  p_{X},\rho_{X}^{AB}\right\}  $ whose
states are mutually orthogonal, there exist nonlocal unitary operations
$U^{AB}$ such that the states of the ensemble $\bar{\varepsilon}=\left\{
p_{X},U^{AB}\rho_{X}^{AB}U^{\dag AB}\right\}  $ are LOCC distinguishable. For
an example we consider the ensemble $\varepsilon$ consisting of the following
states:
\begin{align}
&  U\left(  -\theta\right)  \left\vert 0\right\rangle _{A}\left\vert
0\right\rangle _{B},\quad U\left(  -\theta\right)  \left\vert 0\right\rangle
_{A}\left\vert 1\right\rangle _{B},\\
&  U\left(  -\theta\right)  \left\vert 1\right\rangle _{A}\left\vert
0\right\rangle _{B},\quad U\left(  -\theta\right)  \left\vert 1\right\rangle
_{A}\left\vert 1\right\rangle _{B},
\end{align}
where $U\left(  -\theta\right)  =\exp\left\{  -i\theta\sigma_{x}^{A}\sigma
_{x}^{B}\right\}  $. The similar example has appeared in Ref. \cite{wootters}
where the nonlocal measurements is discussed. The states have the same
entanglement
\begin{equation}
H\left(  \cos^{2}\theta\right)  =-\cos^{2}\theta\log_{2}\cos^{2}\theta
-\sin^{2}\theta\log_{2}\sin^{2}\theta.
\end{equation}
We first consider the case that the states have the equal probability. In this
case the upper bounds obtained from theorem 1 can be expressed as $N\left(
\varepsilon\right)  \leq1$, which is not satisfactory. Note that if Alice and
Bob implement $U\left(  \theta\right)  $ on $AB$ the changed states will be
LOCC distinguishable. The operation $U\left(  \theta\right)  $ can be
implemented through LOCC plus some entanglement. The average entanglement
$\bar{E}\left(  \theta\right)  $ needed to implement $U\left(  \theta\right)
$ is an upper bound of the entanglement charge $N\left(  \varepsilon\right)
$, i.e., $N\left(  \varepsilon\right)  \leq\bar{E}\left(  \theta\right)  $.
Several expressions for $\bar{E}\left(  \theta\right)  $ are given
\cite{wootters,ye2,cirac,berry}, and the one given in Ref. \cite{ye2} shows
that $\bar{E}\left(  \theta\right)  $ will be smaller than unit when
$2\theta\leq0.75$. It means when $2\theta\leq0.75$ the upper bound expression
for $N\left(  \varepsilon\right)  $ obtained by calculating the average
entanglement to implement $U\left(  \theta\right)  $ will be better than that
in theorem 1. However this may not be true when the states have different
probabilities $p_{X}$. To see this we note that $S\left(  \rho^{A}\right)
\geq\sum_{X}p_{X}S\left(  \rho_{X}^{A}\right)  =H\left(  \cos^{2}%
\theta\right)  $, so from theorem 1 there is%
\begin{equation}
N\left(  \varepsilon\right)  \leq S\left(  \rho^{AB}\right)  -S\left(
\rho^{A}\right)  \leq H\left(  X\right)  -H\left(  \cos^{2}\theta\right)  .
\end{equation}
The above upper bound for $N\left(  \varepsilon\right)  $ depends on the
probabilities $p_{X}$ and surely it will be smaller than $\bar{E}\left(
\theta\right)  $ when $H\left(  X\right)  $ is smaller than $\bar{E}\left(
\theta\right)  +H\left(  \cos^{2}\theta\right)  $.

\textit{Discussion.}--Ensembles consisting of mutually orthogonal bipartite
states can be classified into three categories according to the value of the
entanglement charge $N\left(  \varepsilon\right)  $: one has the information
nonlocality, one has the entanglement nonlocality and the third has neither.
For an ensemble, if it has the information nonlocality obviously the states in
it are LOCC indistinguishable, however the inverse may not be true. The reason
is that the value of the entanglement charge $N\left(  \varepsilon\right)  $
depends not only on the states in the ensemble $\varepsilon$ but also on the
probabilities of the states, while whether the states in the ensemble
$\varepsilon$ are LOCC indistinguishable does not depend on their probabilities.

The concepts of the information nonlocality and the entanglement nonlocality
can be extended to ensembles consisting of general bipartite quantum states
since the definition of the entanglement charge $N\left(  \varepsilon\right)
$ does not depend on the orthogonality of the states. The extension is
straightforward and also the entanglement charge $N\left(  \varepsilon\right)
$ can be used as a measure for them. The upper bound expressions in theorem 1
are also applicable for general ensembles but the lower bound expression in
theorem 2 needs to be changed into%
\begin{equation}
N\left(  \varepsilon\right)  \geq\sum p_{X}S\left(  \rho_{X}^{A}\right)
-I_{\rho^{AB}}\left(  A;B\right)  -\Delta\left(  \varepsilon\right)  ,
\end{equation}
where $\Delta\left(  \varepsilon\right)  =S\left(  \rho^{AB}\right)
-I^{Global}\left(  \varepsilon\right)  $, which can be obtained in the same
way as the states are mutually orthogonal.

When the states in the ensemble $\varepsilon=\left\{  p_{X},\rho_{X}%
^{AB}\right\}  $ are not mutually orthogonal, they cannot be distinguished
even through joint measurements and the mutual information $I^{Global}\left(
\varepsilon\right)  $ is the maximal information about the value of $X$ that
can be achieved through physical measurements. Nevertheless, the entanglement
charge $N\left(  \varepsilon\right)  $ still has its operational meaning. The
positive $N\left(  \varepsilon\right)  $ just quantifies the minimal
entanglement that is needed asymptotically in addition to LOCC to achieve
$I^{Global}\left(  \varepsilon\right)  $. When $N\left(  \varepsilon\right)  $
is negative, asymptotically LOCC can get the information $I^{Global}\left(
\varepsilon\right)  $, and additionally at most $\left\vert N\left(
\varepsilon\right)  \right\vert $ ebits of entanglement can be distilled.

\textit{Conclusion.}--We have introduced the entanglement charge as a measure
to quantify nonlocalities in ensembles consisting of bipartite quantum states.
We have estimated various upper and lower bounds for the entanglement charge
and evaluated the exact values for ensembles consisting of mutually orthogonal
maximally entangled bipartite pure states. The present work is expected to
evoke more profound understandings of nonlocalities in ensembles.

\textit{Acknowledgments.}--The work was supported by the RGC of Hong Kong
under HKU7051/06P, HKU7044/08P, the fund of Hebei Normal University, the
Foundation for Universities in Fujian Province (No. 2007F5041), and NSF-China
(Nos. 60878059, 10947147 and 10905016), and theNational Basic Research Program
of China (No. 2006CB921800).


\begin{thebibliography}{99}
\bibitem {add}Notably, a single quantum state may also be regarded as a
special ensemble, and $N\left(  \varepsilon\right)  =-D\left(  \varepsilon
\right)  $ in this case, where $D\left(  \varepsilon\right)  $ denotes the
distillable entanglement of the state.
\end{thebibliography}
\end{document}